\documentclass[copyright,creativecommons]{eptcs}
\providecommand{\event}{WRS 2009} 
\providecommand{\volume}{??}
\providecommand{\anno}{20??}
\providecommand{\firstpage}{1}
\providecommand{\eid}{??}

\usepackage{amsmath,  amssymb}
\usepackage{xspace}
\usepackage{color}
\usepackage{prooftree}
\usepackage{theorem}

\newcommand{\ignore}[1]{}

\theoremstyle{plain}
\newtheorem{lemma}{Lemma}[subsection]
\newtheorem{proposition}[lemma]{Proposition}
\newtheorem{theorem}[lemma]{Theorem}
\newtheorem{definition}[lemma]{Definition}

\newtheorem{corollary}[lemma]{Corollary}

\newtheorem{remark}[lemma]{Remark}
\newenvironment{proof}{\begin{trivlist}\item[]\hspace{\parindent}\textit{ Proof.}}
                      {\hfill \rule{.2em}{.2em}\end{trivlist}}

\newcommand{\regla}[3]{
\begin{array}{c}
  \prooftree #1
    \justifies #2
    \thickness=0.05em
    \using #3
  \endprooftree
\end{array}}

\newcommand{\dummy}{\ensuremath{\,\!}}

\newcommand{\reglaP}[3]{\regla{#1}{#2}{\laRegla{#3}}}

\newcommand{\conc}{\ensuremath{\cdot}}

\newcommand{\laRegla}[1]{\ensuremath{\mathsf{#1}}}

\newcommand{\imp}{\ensuremath{\Rightarrow}\xspace}

\newcommand{\lam}[2]{\ensuremath{\lambda #1 . #2}}
\newcommand{\lamp}[2]{\ensuremath{(\lam{#1}{#2})}} 
\newcommand{\lamm}[3]{\ensuremath{\lam{^{#1}#2}{#3}}}
\newcommand{\lamStar}[2]{\lamm{\star}{#1}{#2}}
\newcommand{\lamma}[2]{\lamm{a}{#1}{#2}}

\newcommand{\appm}[3]{@({#2}^{#1},{#3})}
\newcommand{\appma}[2]{\appm{a}{#1}{#2}}

\newcommand{\setIntersection}{\ensuremath{\cap}}

\newcommand{\sust}[2]{\ensuremath{\{#1 := #2\}}}
\newcommand{\sustLabel}[1]{\ensuremath{\dummy_{\{-#1\}}}}

\newcommand{\twoheadRightarrow}{\ensuremath{\;\Longrightarrow\!\!\!\!\!\!\!\!\Rightarrow}}

\newcommand{\superdevelops}[1]{\ensuremath{\overset{#1}{\Rightarrow}}}
\newcommand{\superdevet}[1]{\ensuremath{\overset{#1}{\Rightarrow_{\et}}}}
\newcommand{\superdevett}[1]{\ensuremath{\overset{#1}{\twoheadRightarrow_{\et}}}}

\newcommand{\est}[1]{\dummy^*_{#1}}

\newcommand{\posAt}[2]{\ensuremath{#1|_{#2}}}
\newcommand{\posRep}[3]{\ensuremath{#1[#2]_{#3}}}
\newcommand{\pos}[1]{\ensuremath{\textsf{pos}(#1)}}

\newcommand{\bPath}[1]{\ensuremath{\textsf{bp}(#1)}}

\newcommand{\disj}{\ensuremath{\uparrow}}

\newcommand{\freeVars}[1]{\ensuremath{\textsf{fv}(#1)}}
\newcommand{\boundVars}[1]{\ensuremath{\textsf{bv}(#1)}}

\newcommand{\terms}{\ensuremath{\Lambda}\xspace}
\newcommand{\termsEtiq}{\ensuremath{\Lambda\et}\xspace}
\newcommand{\et}{\ensuremath{\dummy_{\ell}}}
\newcommand{\sinEtiq}[1]{\ensuremath{|#1|}}

\newcommand{\conjVars}{\mathcal{V}}
\newcommand{\conjEtiq}{\ensuremath{\mathcal{L}}}

\newcommand{\red}{\ensuremath{\rightarrow}}
\newcommand{\weak}{\ensuremath{_{\ell}}}

\newcommand{\wred}{\ensuremath{\rightarrow\weak}}
\newcommand{\wredDisj}[1]{\ensuremath{\overset{#1}{\rightarrow}\weak}}
\newcommand{\wreddDisj}[1]{\ensuremath{\overset{#1}{\twoheadrightarrow}\weak}}

\newcommand{\wredd}{\ensuremath{\twoheadrightarrow\weak}}

\newcommand{\redEnc}[1]{\ensuremath{\overset{#1}{\rightsquigarrow}}}

\newlength{\widebarargwidth}
\newlength{\widebarargheight}
\newlength{\widebarargdepth}

\newcommand{\rred}[1]{\ensuremath{\twoheadrightarrow_{#1}}}
\newcommand{\dev}{\ensuremath{\Rightarrow}}
\newcommand{\igdef}{\ensuremath{\circeq}}

\newcommand{\contexto}[2]{\ensuremath{\langle#1, #2\rangle}}
\newcommand{\reduceRedex}[1]{\ensuremath{\overset{#1}{\red}}}

\newcommand{\termsMarcados}{\ensuremath{\Lambda^{\star}}}

\newcommand{\LC}{$\lambda$-calculus\xspace}
\newcommand{\weakLC}{$\lambda^w$-calculus\xspace}
\newcommand{\weakLCM}{$\lambda^w_{\star}$-calculus\xspace}
\newcommand{\weakLCL}{$\lambda^w_{\ell}$-calculus\xspace}

\newcommand{\fl}[1]{\ensuremath{\textsf{fl}(#1)}}

\title{Superdevelopments for Weak Reduction\thanks{Work partially supported by Instituto Tecnol\'ogico de Buenos Aires and LIFIA}}

\author{Eduardo Bonelli
\institute{CONICET and Universidad Nacional de Quilmes (Argentina)}
\email{ebonelli@unq.edu.ar}
\and
Pablo Barenbaum
\institute{Universidad de Buenos Aires (Argentina)
}
\email{\quad foones@gmail.com}
}

\def\titlerunning{Superdevelopments for Weak Reduction}
\def\authorrunning{E. Bonelli \& P. Barenbaum}
\begin{document}
\maketitle

\begin{abstract}
We study superdevelopments in the weak lambda calculus of \c{C}a\v{g}man and
Hindley, a confluent variant of the standard weak lambda calculus in
which reduction below lambdas is forbidden.  In contrast to
developments, a superdevelopment from a term $M$ allows not only
residuals of redexes in $M$ to be reduced but also some newly created
ones. In the lambda calculus there are three ways new redexes may be
created; in the weak lambda calculus a new form of redex creation is
possible. We present labeled and simultaneous reduction formulations of
superdevelopments for the weak lambda calculus and prove them
equivalent.
\end{abstract}



\section{Introduction}


In contrast to \LC, which allows reduction under the lambda, weak \LC
does not. This results in a calculus which is arguably more relevant
to programming languages given that the latter consider abstractions
as values. However, simply dropping the reduction scheme:
\begin{center}
$\regla{M\reduceRedex{}N}
       {\lam{x}{M}\red{}\lam{x}{N}}
       {\xi}
$
\end{center}
causes confluence to fail, as may be easily
verified. A restriction of the $\xi$-scheme recovers
confluence~\cite{CagmanHindley:1998,LevyMaranget:1999}. Here, the
judgement $M\reduceRedex{\Delta}N$ means ``$M$ reduces to $N$ by
contracting redex $\Delta$'' and $\freeVars{\Delta}$ are the free
variables of $\Delta$:
\begin{center}
$\regla{M\reduceRedex{\Delta}N
        \quad x\notin\freeVars{\Delta}}
       {\lam{x}{M}\reduceRedex{\Delta}\lam{x}{N}}
       {\xi'}
$\end{center}

The resulting weak $\lambda$ calculus (\weakLC) enjoys \textit{finite
  developments}: all developments are finite and end in the same
term. A development from a term $M$ is a reduction sequence in which
only residuals of redexes present in $M$ are reduced. In this paper we
study \textit{superdevelopments}~\cite{Aczel:1978,vanRaamsdonk:1993}
in the \weakLC.
A superdevelopment allows not only residuals of redexes in $M$ to be
reduced but also those upward created ones in the sense that the created redex occurs at a prefix of the reduced redex. There are three ways in
which a redex may be created in \LC~\cite{TesisLevy:1978}:
\begin{itemize}
\item [I.] $(\lam{x}{x})\,(\lam{y}{P})\,Q\reduceRedex{} (\lam{y}{P})\,Q$

\item [II.] $(\lam{x}{\lam{y}{P}})\,R\,Q\reduceRedex{} (\lam{y}{P\sust{x}{R}})\,Q$

\item [III.] $(\lam{x}{C[x\,Q]})\,\lam{y}{P}\reduceRedex{} C'[(\lam{y}{P})\,Q']$, where $C'=C\sust{x}{\lam{y}{P}}$ and $Q'=Q\sust{x}{\lam{y}{P}}$.

\end{itemize}
A superdevelopment from $M$ allows contraction of newly created
redexes of type I and II (i.e. upward creation). In \weakLC we meet two
differences. First, redex creation of type III is restricted to those
cases in which $Q$ does not have free occurrences of variables bound
in $C$ above the hole. Second, there is  a new way of creating
redexes (hence redex creation in \weakLC is not derived from that of \LC):
\begin{itemize}
\item [IV.] $(\lam{x}{C[(\lam{y}{P})\,Q]})R\,\reduceRedex{}
  C'[(\lam{y}{P'})\,Q']$, where $x\in \freeVars{(\lam{y}{P})\,Q}$
  and no free variables in $(\lam{y}{P})\,Q$ are bound in $C$;
  $C'=C\sust{x}{R}$ and $Q'=Q\sust{x}{R}$ and $P'=P\sust{x}{R}$.

\end{itemize}
In the reduction step $(\lam{x}{I\,x})\,y\reduceRedex{} \underline{I\,y}$,
where $I=\lam{x}{x}$, the underlined redex is a new redex of type IV.

We define weak superdevelopments (i.e. superdevelopments in \weakLC)
by means of an appropriate labeling scheme. Although attractive due to
its conciseness, this definition requires reducing terms to
\textit{normal form} by means of this notion of labeled
reduction. More convenient is a direct inductive
defi\-ni\-tion. Therefore, we introduce a notion of simultaneous
reduction, similar to simultaneous (a.k.a. parallel) reduction in
\LC. The topic of this paper is to exhibit the complications arising
when one tries to prove these notions equivalent and the approach we
take to resolve these issues.

\ignore{

We then address higher-order rewriting. First a notion of orthogonal
higher-order rewrite system is introduced which is suited for weak
reduction. For example, the following PRS is orthogonal for
weak reduction whereas it is {\em not} for strong reduction:
\begin{center}
$\begin{array}{rcl}
g(\lam{x}{f\,x}) & \red{} & a\\
f\,y & \red{} & b
\end{array}$
\end{center}
The notion of superdevelopment for weak reduction in higher-order
rewriting is more subtle than in the lambda calculus. When extending (\ref{eqn:superDev}) two issues arise. 

[Tema de identificar varibales que se sustituyen]

We show why the naive extension fails and motivate the need for
tracking how redexes contribute to the creation of other redexes. In
particular, in the setting of \weakLC, it turns out that all redexes
involving $x$ that are reduced in the judgement $M\dev^L \lam{x}{M'}$
of  do not contribute to the creation of the head
redex. In higher-order rewriting this may not be the case. In other
words, not all redexes involving

[Tema de no permitir todos los redexes que involucren x]

} 

\textbf{Motivation.} The starting point of this work is an attempt at
extending the concept of orthogonal systems and the confluence results
of Mayr and Nipkow~\cite{MayrNipkow:1998} to weak higher-order
rewriting. Orthogonality depends on whether weak or strong reduction
is considered. For eg. in weak reduction $\beta\eta$ as an
HRS~\cite{MayrNipkow:1998} has only one critical pair (in contrast to
the two critical pairs that arise under strong reduction). Moreover,
some systems such as $\{ g(\lambda x.f(x)) \red a, f(y)\red b\}$ are
not orthogonal for strong reduction but are for weak
reduction. Another difference between weak and strong reduction lies
in standardization~\cite{TERESE}. Standardization results for strong
reduction in higher-order rewriting, such as HRS, apply to
left-linear, pattern HRS that are \textit{fully-extended} (roughly
that redexes behave as in the first-order case - they are determined
exclusively by their term structure). For instance, the first rule in
the HRS $\{ f(\lambda x.y) \red g(y), h(x)\red a \}$ is not fully
extended since in order to determine if a term of the form $f(\lambda
x.M)$ is a redex, it must be checked whether $x$ occurs free in $M$ or
not. Fully-extendedness is needed in order for anti-standard pairs to
be swappable and hence non-standard reductions to be standardized. For
e.g. in $f(\lambda x.h(x)) \red f(\lambda x.a)\red g(a) $ the first
and the second steps cannot be swapped. Note, however, that in weak
reduction fully-extendedness is not required: indeed, the first
rewrite step in the derivation is not a valid weak step. Returning to
our initial motivation on confluence for orthogonal weak HRS, the
proof of Mayr and Nipkow~\cite{MayrNipkow:1998} resorts to Aczel's
notion of superdevelopments (but for HRS). As will be explained in
this work, a formalization of weak superdevelopments in terms of
simultaneous reduction for HRS is involved. The reason is that it
requires allowing some (but not all) redexes to be reduced
\textit{under} abstractions. For eg. in the HRS $\{ f (\lam{x}{g(y\,
  x)}) \red{} y\,a, k(x)\red{} x \}$, the term $f (\lam{x}{g(k(x))})$
weakly superdevelops to $a$ (note that the redex $k(x)$ is allowed to
be reduced), however $f (\lam{x}{k(g(x))})$ does not weakly
superdevelop to $a$ (the redex $k(g(x))$ is not allowed to be
reduced). Therefore, in order to get a clear grasp of the problem we
choose to first address this task for the \weakLC. It should be
mentioned that proofs of confluence for weak, orthogonal HRS that rely
on developments rather than superdevelopments should go
through. However, a number of applications of superdevelopments to
confluence, normalisation and higher-order matching, as discussed
below, suggest that the concept of superdevelopment deserves to be
studied in its own right.

\textbf{Related work.}  
According to \c{C}a\v{g}man and
Hindley~\cite{CagmanHindley:1998} weak reduction, as presented in this
work, is due to Howard~\cite{Howard:1968}. It arises as an attempt to
construct a tighter correspondence between reduction in Combinatory
Logic and $\beta$ reduction. \c{C}a\v{g}man and
Hindley~\cite{CagmanHindley:1998} give a clear account of \weakLC\ and
its relation to $\beta$ reduction. L\'evy and
Maranget~\cite{LevyMaranget:1999} study a calculus of explicit
substitutions for \weakLC, stating that the theory of weak reduction
is rather poorly developed in the literature. In a sequel
paper~\cite{BLM:2005} they introduce a labeled variant in order to
study sharing for this calculus.
Notions of weak reduction for calculi with explicit substitutions have
been studied by Fern\'andez et al~\cite{FMS05a,FMS05b}. The latter
considers two variants of the $\mathit{Beta}$ rule, one in which the
argument is required to be closed and one in which the function is
required to be closed. This suggests other variations on weak
reduction for HRS, although implementation concerns for HRS are out of
the scope of this work.
Superdevelopments were introduced by Aczel to prove confluence of a
generalization of lambda calculus~\cite{Aczel:1978}. Van
Raamsdonk~\cite{vanRaamsdonk:1993} proves finiteness of
superdevelopments and confluence of orthogonal CRS by adapting Aczel's
technique. Two additional different proofs of finiteness are given
in~\cite{TesisVanRaamsdonk:1996}.  Mayr and Nipkow use a similar
technique to prove confluence of orthogonal
PRS~\cite{MayrNipkow:1998}.  Another application of superdevelopments
is in higher-order matching. This problem is usually stated in the
setting of typed lambda calculus. In order to obtain decidable
subclasses of this problem, terms are usually restricted to some
particular order. An alternative approach in restricting the problem
is to weaken the reduction relation from reduction to normal form to
superdevelopments. de Moor and
Sittampalam~\cite{MS:1998,MS:2001,SM:2001,Faure:2006} study matching
modulo superdevelopments. Khasidashvili and Piperno~\cite{KP:1998}
show that the amount of superdevelopments required to normalize
certain classes of terms can be determined statically.

\textbf{Preliminaries.} Assume given a denumerably infinite set of term variables $\conjVars$.
The set of (\weakLC) terms \terms and contexts are defined as follows:
\begin{center}
$M  ::=  x\, |\, M\, M\, |\, \lam{x}{M}$\hspace{1cm}
$C  ::=  \Box\, |\, C\,M\, |\, M\, C\, |\, \lam{x}{C}$
\end{center}
 Free (\freeVars{M}) and bound (\boundVars{M}) variables are as usual.
 We assume the convention that bound variables are different from free
 variables and, moreover, bound variables of distinct binders have
 been renamed apart.  Capture avoiding substitution of all free
 occurrences of $x$ in $M$ by $N$ is written $M\sust{x}{N}$. In a
 statement in which distinct variables $x_1, \hdots, x_n$ occur we use
 $\overline{x}$ for the set $\{x_1, \hdots, x_n\}$.  We write $C[M]$
 for the result of replacing the hole in $C$ with $M$ (this may bind
 the free variables of $M$ in $C[M]$). The binding path of $C$,
 $\bPath{C}$, is the sequence of variables that are bound in $C$ above
 the hole ($\igdef$ is definitional equality): $\bPath{\Box}\igdef
 \epsilon$, $\bPath{C\, N}\igdef\bPath{C}$, $\bPath{M\,
   C}\igdef\bPath{C}$ and $\bPath{\lam{x}{C}}\igdef\bPath{C}\conc
 x$. A position ($p,q,r$) is a sequence of positive integers;
 $\epsilon$ is the root position (empty sequence) and $p\conc q$ is
 the (associative) operation of sequence composition. If $P$ is a set
 of positions we write $p\conc P$ for the set resulting from composing
 $p$ with each position in $P$.  We write \pos{M} for the set of
 positions of $M$: $\pos{x} \igdef \{\epsilon\}$, $\pos{M\, N} \igdef
 (0\conc \pos{M})\cup (1\conc\pos{N})$, $\pos{\lam{x}{M}} \igdef
 1\conc \pos{M}$. The subterm of $M$ at position $p$ is \posAt{M}{p}.
 Also, \posRep{M}{N}{p} stands for the term resulting from replacing
 the subterm at position $p$ in $M$ with $N$ (this may bind the free
 variables of $N$ in \posRep{M}{N}{p}).  We write \rred{} for the
 reflexive--transitive closure of a binary relation \red{}.  We write
 $S,T$ for sequences of variables and $U,V$ for sets of
 variables. Also, $S \oplus T$ is the sequence resulting from
 concatenating $S$ with $T$. We say a term $M$ is \textit{away from} a
 sequence of variables $S$ and write $M \disj S$ if $\freeVars{M}$ is
 disjoint from $S$.  We write $S\subseteq T$ to indicate that the
 underlying set of $S$ is included in that of $T$.

\textbf{Structure of the paper.} Sec.~\ref{sec:redexCreation}
proves that the above mentioned redex creation types are the only
possible ones. Sec.~\ref{sec:definingSuperdevelopments} introduces two
definitions of weak superdevelopments in \weakLC. Sec.~\ref{sec:equivalence} addresses the proof of equivalence of these. Finally, we
conclude and report on our ongoing work on extensions to higher-order
rewriting.



\section{Redex Creation in \weakLC}
\label{sec:redexCreation}


This section characterizes redex creation in \weakLC.  In order to
follow redexes along reductions we mark them (with a star) and only
allow such marked redexes to be contracted. The set of terms
\termsMarcados\ and contexts of the resulting formalism (\weakLCM) are
defined as follows:
\begin{eqnarray*}
M & ::= & x\, |\, M\, M\, |\, \lam{x}{M}\, |\, (\lamm{\star}{x}{M})\,M\\
C & ::= & \Box\, |\, C\, M\, |\, M\, C\, |\, \lam{x}{C}\, |\, (\lamm{\star}{x}{C})\,M\, |\, (\lamm{\star}{x}{M})\,C
\end{eqnarray*}
The set of positions and binding path is extended accordingly. In the
sequel of this subsection, when we speak of ``terms'' we mean ``marked
terms'' and likewise for contexts. If $M$ is a term and $p \in
\pos{M}$, then we write $\contexto{M}{p}$ for the context resulting
from replacing the term at $p$ in $M$ with a hole.  If the
hole in $C$ is at position $p$ we write $C[\,]_p$.

As mentioned, reduction in \weakLCM is similar to reduction in \weakLC except only marked redexes may be contracted.

\begin{center}
$\begin{array}{c}
\reglaP{\Delta = (\lamStar{x}{M})\,N}
       {(\lamStar{x}{M})\,N \reduceRedex{\Delta} M \sust{x}{N}}
       {}
\hspace{1cm}
\reglaP{M \reduceRedex{\Delta} M'
        \quad
        \Delta\disj\bPath{C}}
       {C[M] \reduceRedex{\Delta} C[M']}
       {}
\end{array}
$\end{center} 

We are interested in studying situations where reduction in \weakLCM,
from terms where all redexes have been marked, produces terms having
occurrences of $(\lam{x}{M})\,N$ in contexts with binding paths that
they are away from. This models the situation where a new reducible
expression, a \weakLC\ redex, has been produced that was not initially
marked.

\begin{definition}
Let $M,N,P \in \termsMarcados$, $p \in \pos{M}$ and $\posAt{M}{p} \disj\bPath{\contexto{M}{p}}$. 
\begin{itemize}
\item If $\posAt{M}{p} = (\lamStar{x}{P})\,N$, then we say \posAt{M}{p}
is a {\em \weakLCM\ redex at $(M,p)$}.

\item If $\posAt{M}{p} = \lamp{x}{P}N$, then we say \posAt{M}{p} is a
{\em \weakLC\ redex at $(M, p)$}.
\end{itemize}
\end{definition}

A term $M \in \termsMarcados$ is {\em initially marked} iff the
set of marked subterms are indeed \weakLCM\ redexes and all \weakLC\ redexes have
been marked. More precisely, iff the following conditions hold:

\begin{enumerate}
\item $\forall p \in \pos{M}.\posAt{M}{p} = (\lamStar{x}{P})\,Q$ implies
  \posAt{M}{p} is a \weakLCM\ redex at $(M, p)$.
 
\item $\forall p \in \pos{M}.\posAt{M}{p} = \lamp{x}{P}\,Q$ implies
  \posAt{M}{p} is not a \weakLC\ redex at $(M, p)$.

\end{enumerate}

The following result is proved by case analysis on the relative positions of $p$ and $q$.

\begin{proposition}[Redex creation]
\label{prop:onRedexCreation}
Let $M \in \termsMarcados$ be initially marked, $M
\reduceRedex{\Delta} N$ and $q \in \pos{N}$ s.t. there is a
\weakLC\ redex at $(N, q)$.  One of the following situations must
arise:
\begin{itemize}
\item \textbf{Case I}: $M = C[(\lamStar{x}{x})\,\lamp{y}{M_1}M_2]$ and $N =
  C[\lamp{y}{M_1}\,M_2]_q$.

\item \textbf{Case II}: $M = C[(\lamStar{x}{\lamp{y}{M_1}})\, Q\, M_2]$ and $N =
  C[\lamp{y}{M_1'}M_2]_q$, where $M_1' = M_1\sust{x}{Q}$.

\item \textbf{Case III}: $M = C_1[(\lamStar{x}{C_2[x
      M_2]})\,\lam{y}{M_1}]$ and $N =
  C_1[C_2'[\lamp{y}{M_1}M_2']_{q_2}]_{q_1}$, where $q=q_1\conc q_2$, $C_2'
  = C_2\sust{x}{\lam{y}{M_1}}$ and $M_2' = B\sust{x}{\lam{y}{M_1}}$.

\item \textbf{Case IV}: $M = C_1[(\lamStar{x}{C_2[\lamp{y}{M_1}M_2]})\,Q]$
  and $N = C_1[C_2'[\lamp{y}{M_1'}M_2']_{q_2}]_{q_1}$, where $q=q_1\conc
  q_2$, $M_1' = M_1\sust{x}{Q}$, $M_2' = M_2\sust{x}{Q}$, $C_2' =
  C_2\sust{x}{Q}$ and $x \in \freeVars{\lamp{y}{M_1'}M_2'}$.

\end{itemize}
\end{proposition}




\section{Superdevelopments in \weakLC}
\label{sec:definingSuperdevelopments}

This section introduces two presentations of superdevelopments in
\weakLC: via labeled reduction
(Sec.~\ref{subsec:developmentsViaLabeledReduction}) and simultaneous
reduction (Sec.~\ref{subsec:developmentsViaParallelReduction}).


\subsection{Weak Superdevelopments via Labeled Reduction}
\label{subsec:developmentsViaLabeledReduction}

We begin by introducing the labeled \weakLC (\weakLCL).  Assume given
a denumerably infinite set of {\em labels} $\conjEtiq$. The labeled
terms \termsEtiq and contexts are given by the following grammar:
\begin{center}
$A  ::=  x\, |\, \lamma{x}{A}\, |\, \appma{A}{A}$\hspace{1cm}
$C  ::=  \Box\, |\, \lamma{x}{C}\, |\, \appma{C}{A}\, |\, \appma{A}{C}$
\end{center}
where $x \in \conjVars$ and $a \in \conjEtiq$.
Ocassionally, we write $\lamm{a_1 \hdots a_n}{x_1 \hdots
    x_n}{A}$ (or simply $\lamm{\overline{a}^n}{\overline{x}^n}{A}$) as
  a shorthand for
  $\lamm{a_1}{x_1}{\lamm{a_2}{x_2}{\hdots\lamm{a_n}{x_n}{A}}}$.
 If we wish to single out the leftmost abstraction we write
  $\lamm{a\,\overline{a}^n}{x\,\overline{x}^n}{A}$.   
Thus abstractions and (the first argument of) applications are decorated with labels. In $\appma{A}{B}$ the depicted label binds all the occurrences of $a$ in $A$.  The set of free labels of a labeled term
  is defined as follows:
    \begin{eqnarray*}
    \fl{x} & \igdef & \emptyset\\
    \fl{\lamma{x}{A}} & \igdef & \{a\} \cup \fl{A}\\
    \fl{\appma{A}{B}} & \igdef & \fl{A}\setminus \{a\}\cup\fl{B}
    \end{eqnarray*}
We assume the existence of a distinguished label
$\star\in\conjEtiq$ which is never bound. Also, $M\sustLabel{a}$
denotes the substitution of all free occurrences of label $a$ with
$\star$, and $\sinEtiq{A} \in \terms$ is the term resulting from $A$
by erasing all labels (and identifying $\appm{}{A}{B}$ with $A\,B$),
in which case we say $A$ is a {\em labeling} of
$\sinEtiq{A}$. Substitution over labeled terms is written
$A\sust{x}{B}$. Note that this operation must not capture labels. For
example, $\appm{a}{x}{x}\sust{x}{\lamma{y}{y}} \neq
\appm{a}{(\lamma{y}{y})}{(\lamma{y}{y})}$, rather the binding
occurrence of $a$ in $\appm{a}{x}{x}$ must first be renamed to
$\appm{b}{x}{x}$ in order for substitution to yield
$\appm{b}{x}{x}\sust{x}{\lamma{y}{y}}=
\appm{b}{(\lamma{y}{y})}{(\lamma{y}{y})}$. The binding path of a
labeled context $C$ is the binding path of \sinEtiq{C}.

\begin{remark}
Labeled characterizations of superdevelopments we know of do not bind
labels in applications. Instead they define a term to be
\textit{good}~\cite{vanRaamsdonk:1993}\footnote{This work decorates
  applications with \textit{sets} of labels, however the example of
  this remark still applies by considering $\{a\}$ instead of $a$
  where appropriate.} or
\textit{well-labeled}~\cite{TesisVanRaamsdonk:1996,Faure:2006} if an
occurrence of a label $a$ decorates an application and another
occurrence of $a$ decorates an abstraction, then this abstraction must
occur in one of the arguments of the application. It is stated that
reduction preserves well-labeledness, however this in fact
fails. Eg. taking $A\igdef \appm{a}{(\lamm{a}{y}{y})}{z}$, clearly the
reduction step $\appm{c}{(\lamm{c}{x}{\appm{b}{x}{x}})}{A}\red{}
\appm{b}{A}{A}$ produces a non well-labeled term. Note that the
results in op.cit. still hold (except for preservation of
well-labeledness under reduction, as illustrated) since, by
well-labeledness, these copies of $A$ never interact with one another.
\end{remark}


A term $A$ is said to be {\em{initially labeled}} iff
all abstractions have distinct labels.
    Note that, given a term $A$, it
    is always possible to produce a label $a$ that does not occur in it
    given that $A$ is finite and $\conjEtiq$ is not.  
Reduction in \weakLCL is
defined below, where $S$ is a sequence of variables and the depicted
occurrence of $\appma{(\lamma{x}{A})}{B}$ is called a {\em redex}.
\begin{center}
$\regla{\appma{(\lamma{x}{A})}{B} \disj \bPath{C}\oplus S}
       {C[\appma{(\lamma{x}{A})}{B}] \wredDisj{S} C[A\sustLabel{a}\sust{x}{B}]}
       {}
$\end{center}

We substitute $a$ with $\star$ to avoid bound labels from becoming
free as in $\appma{(\lamma{x}{\lamma{y}{y}})}{w} \wredDisj{S}
\lamma{y}{y}$ and, consequently to avoid rebinding
of labels such as in
$\appma{\underline{\appma{(\lamma{x}{\lamma{y}{y}})}{w}}}{z}\wredDisj{S}
\appma{\lamma{y}{y}}{z}$. We write $\wred$ for $\wredDisj{\emptyset}$.
Note that $A \wredDisj{S} B$ and $T \subseteq S$ implies $A
\wredDisj{T} B$. Also, reduction does not create free variables. In
the judgement $A \wredDisj{S} B$ we implicitly assume that $A,B\in
\termsEtiq$.

%

\begin{definition}
$M$ {\em weakly superdevelops to} $N$ if there exists an
initially labeled term $A$ and a labeled term $B$ s.t.
$\sinEtiq{A}=M$ and $\sinEtiq{B}=N$ and $A\wredd B$.  If, moreover,
$B$ is in $\wred$-normal form, then we say there is a {\em complete}
weak superdevelopment from $M$ to $N$.
\end{definition}

Some properties of labeled reduction are considered below. Since weak
superdevelopments are also superdevelopments (which are
finite~\cite{vanRaamsdonk:1993,TesisVanRaamsdonk:1996}):

\begin{lemma}
Weak superdevelopments are finite.
\end{lemma}

A further property is that $S$ may be weakened with further variables
not occurring free in $A$ without affecting reduction. This extends to
many-step reduction.

\begin{lemma}
\label{lem:weakeningOfBindingPathPreservesReduction}
$A \wredDisj{S} B$ and $x \not\in \freeVars{A}$ implies
$A \wredDisj{S \oplus x} B$.
\end{lemma}

Regarding reduction and the context in which the redex occurs:

\begin{lemma}
\label{lemma:auxSubstitutionPreservesReduction}

The following are equivalent:
\begin{enumerate}
\item Suppose $A \wred B$ by contracting a redex
  $\Delta$. $\Delta \disj S$ iff $C[A] \wred C[B]$ for
  every context $C$ s.t. $\bPath{C} = S$.

\item $C_1[C_2[\Delta]] \wredDisj{S} C_1[C_2[\Delta']]$ iff $
  C_2[\Delta] \wredDisj{S \oplus \bPath{C_1}} C_2[\Delta']$, where
  $\Delta$ is the contracted redex.

\end{enumerate}
\end{lemma}

Finally, we prove that substitution preserves reduction
(Lem.~\ref{lemaReduccionDeSustituciones}(2)). The proof is by
induction on $A$, resorting to
Lem.~\ref{lemma:auxSubstitutionPreservesReduction}(2) and
Lem.~\ref{lemaReduccionDeSustituciones}(1).

\begin{lemma}
\label{lemaReduccionDeSustituciones}
\begin{enumerate}
\item If $a \not\in \fl{B}$,
$A\sustLabel{a}\sust{x}{B}= A\sust{x}{B}\sustLabel{a}$.

\item Suppose $A \wredDisj{S} A'$ and $B\disj S$. Then 
$A\sust{x}{B} \wredDisj{S} A'\sust{x}{B}$.
\end{enumerate}
\end{lemma}

\begin{proof}
Both proofs procceed by induction on $A$.
In the cases where a variable $y$ is bound in $A$, we assume
the convention that $y$ does not occur in $B$. Also, when a
head reduction takes place, we resort to the fact that
$A_1\sust{x}{B}\sust{y}{A_2\sust{x}{B}} = A_1\sust{y}{A_2}\sust{x}{B}$. 

\end{proof}

\subsection{Weak Superdevelopments via Simultaneous Reduction}
\label{subsec:developmentsViaParallelReduction}

An alternative presentation of weak superdevelopments is by means of
simultaneous reduction.
    It has numerous benefits over labeled
    reduction. One is that it satisfies the diamond property (and can be
    used for proving confluence of labeled reduction and the
    \weakLC\ (Prop.~\ref{prop:confluenceOfWeakLC})). Another is that we
    can avoid reasoning over reduction to a normal form: one simultaneous
    reduction step suffices for superdeveloping a term.
A naive attempt at formalizing this notion in big-step style
fails. Let us write $M\superdevelops{S}N$ for such a judgement, where
$S$ is a sequence of variables which denotes the binding context
(i.e. the variables that are bound in the context) in which this
superdevelopment takes place.  This judgement would include the
inference scheme:
\begin{equation}
\regla{M\superdevelops{S} \lam{x}{M'}
          \quad N\superdevelops{S} N'
          \quad (\lam{x}{M'})\,N'\disj S} 
         {M\,N \superdevelops{S} M'\sust{x}{N'}} 
         {}  
\label{eqn:superDev}
\end{equation}
However, it turns out that we need to consider an exception to the
condition $(\lam{x}{M'})\,N'\disj S$, namely when an
abstraction contributes to a head redex in a superdevelopment. Indeed,
in this case the abstracted variable is to be substituted and hence
redexes which contain occurrences of this variable \textit{can} be
contracted.  
As an example, consider the following superdevelopment in
the \weakLC:
\begin{center}
$I\,(\lam{x}{I\,x})\,y\red{} (\lam{x}{I\,x})\,y\red{} I\,y\red{}y$
\end{center}
To deduce the judgement $I\,(\lam{x}{I\,x})\,y\superdevelops{S} y$
using (\ref{eqn:superDev}) we require
$I\,(\lam{x}{I\,x})\superdevelops{S} \lam{x}{x}$.  For this to hold we
must allow contraction of $I\,x$. As a consequence, we study an
extended judgement $M\superdevelops{S,k} N$ in which the integer $k
\geq 0$ indicates how many abstractions of $M$ contribute to a head
redex in a later stage of the derivation. Reduction under these
abstractions is allowed. Note that, in contrast to
Sec.~\ref{subsec:developmentsViaLabeledReduction} where $S$ in
$A\wredDisj{S} B$ includes all the variables in $A$ bound above the
redex that is reduced, in the judgement
$M\superdevelops{S,k} N$ the sequence $S$ contains only the variables
in $M$ bound above the redexes of the superdevelopment under which
reduction is not allowed.


\begin{definition}
There is a {\em superstep from $M$ to $N$  under $S,k$} iff
$M\superdevelops{S,k}N$, where this judgement is defined as follows:
\begin{center}
$\begin{array}{c}
\reglaP{}{x \superdevelops{S,0} x}{Var}
\hspace{.5cm}
\reglaP{M \superdevelops{x \conc S,0} M'}
       {\lam{x}{M} \superdevelops{S,0} \lam{x}{M'}}{Abs1}
\hspace{.5cm}
\reglaP{M \superdevelops{S,k} M'}
       {\lam{x}{M} \superdevelops{S,k+1} \lam{x}{M'}}{Abs2}
\hspace{.5cm}
\reglaP{M \superdevelops{S,0} M'
        \quad 
        N \superdevelops{S,0} N'}
       {M N \superdevelops{S,0} M' N'}{App1}
\\
\\
\reglaP{M \superdevelops{S,n+1} \lam{x}{M'}
        \quad 
        N \superdevelops{S,m} N'
        \quad \lamp{x}{M'}N' \disj S
        \quad m > 0 \imp M' = \lam{\overline{x}^n}{x} }
       {M N \superdevelops{S,n+m} M'\sust{x}{N'}}{App2}
\end{array}$
\end{center}

There is a {\em complete} superstep from $M$ to $N$ under $S,k$ if in
the derivation of the judgement $M \superdevelops{S,k} N$ the scheme
\laRegla{App1} is used only if \laRegla{App2} is not applicable.
 
\end{definition}

Note that this definition establishes an inside-out strategy for
computing a complete weak superdevelopments.



\section{Equivalence of Presentations}
\label{sec:equivalence}

In this section we prove the following result (the first item in
Sec.~\ref{subsec:parallelToLabeled} and the second in
Sec.~\ref{subsec:labeledToParallel}), where we write $M\et$ for any
labeling of $M$.

\begin{theorem}
\label{thm:equivalence}
\begin{enumerate}

\item If $M \superdevelops{S,0} N$, then there exist $M\et,
  N\et$ s.t.  $M\et \wreddDisj{S} N\et$.

\item If $M\et \wreddDisj{S} N\et$  and
  $N\et$ is in normal form, then $M \superdevelops{S,0} N$.

\end{enumerate}
\end{theorem}

In the second item, note that the binding nature of labels in
applications is required for the statement to hold. For example, this
is the reduction sequence one would obtain were labels in applications not
considered binding:
\begin{center}
$\appm{b}{\lamm{b}{x}{\appm{a}{x}{y}}}{\lamma{z}{z}}
\wred 
\appm{a}{(\lamma{z}{z})}{y}
\wred
y$ 
\end{center}
Notice that it is not the case that
$\appm{b}{\lamm{b}{x}{\appm{a}{x}{y}}}{\lamma{z}{z}}
\superdevelops{S,0} y$ as may be seen by trying to derive this
judgement.  
The requirement that $N\et$ be in normal form is justified
by the following example, where $A$ is any redex s.t. $A\wred A'$:
\begin{center}
$\appm{a}{(\lamm{a}{x}{\appm{b}{x}{x}})}{A}\wred \appm{b}{A}{A} \wred \appm{b}{A'}{A}$.
\end{center}
It is clear that labeled reduction still has some work to do: a
labeled redex remains (i.e. it is an incomplete weak
superdevelopment). In fact, this is an example of an incomplete
development.  For this reason, the judgement
$\appma{(\lamma{x}{\appm{b}{x}{x}})}{A} \superdevelops{S,0} A\,A'$ is
not derivable.


\subsection{Supersteps over Labeled Terms}


As mentioned, we have to relate labeled reduction to {\em normal form}
with supersteps. In order to do so, we introduce an intermediate
notion of supersteps over {\em labeled} terms $A \superdevet{S,k}
B$. The reason is that when passing from $\wred$ steps to
$\superdevelops{S,k}$ steps we lose the labels and hence our handle
over this normal form (which is the complete weak superdevelopment of
the {\em labeled} redexes).  In summary, in
Sec.~\ref{subsec:parallelToLabeled}, in relating $\superdevelops{S,k}$
with $\wredd$, we shall first go from $\superdevelops{S,k}$ to
$\superdevet{S,k}$ (defined below) and then to $\wredd$. Conversely,
in Sec.~\ref{subsec:labeledToParallel}, we shall first go from
$\wredd$ to $\superdevet{S,k}$ and then to $\superdevelops{S,k}$.

\begin{definition}[Supersteps over labeled terms]
\label{def:simultaneousLabeledWeakSuperdevs}
We say there is a {\em superstep from $A$
    to $B$ under $S,k$} iff $A \superdevet{S,k} B$, where this
  judgement is defined as follows: 
\begin{center}
$\begin{array}{c}
    \reglaP{}{x \superdevet{S,0} x}{LVar}
    \hspace{0.5cm} 
    \reglaP{A \superdevet{x \conc S,0} A'}
           {\lamma{x}{A}
      \superdevet{S,0} \lamma{x}{A'}}
           {LAbs1}
    \\
    \\
    \reglaP{A \superdevet{S,k} A'}
           {\lamma{x}{A}
      \superdevet{S,k+1} \lamma{x}{A'}}
           {LAbs2}
    \hspace{0.5cm} 
    \reglaP{A \superdevet{S,0} A'
            \quad 
            B \superdevet{S,0} B'}
           {\appma{A}{B} \superdevet{S,0} \appma{A'}{B'}}
           {LApp1}
    \\
    \\
\reglaP{A \superdevet{S,n+1} \lamma{x}{A'}
       \quad
       B \superdevet{S,m} B'
       \quad \appma{(\lamma{x}{A'})}{B'} \disj S
       \quad m > 0 \imp A' = \lamm{\overline{a}^n}{\overline{x}^n}{x} }
       {\appma{A}{B} \superdevet{S,n+m} A'\sustLabel{a}\sust{x}{B'}}
       {LApp2}
\end{array}
$
\end{center}
\end{definition}

Note that $\superdevet{S,k}$ is reflexive and $A \superdevet{S,k} B$
implies $\freeVars{A} \supseteq \freeVars{B}$.  Additional basic
properties of reduction are stated below.
    The first one states
    that a weak superdevelopment under $S,k$ leaves $k$ abstractions at
    the root of the term. The remaining ones indicate how labeled and
    unlabeled simultaneous reduction relate.

\begin{lemma}
\label{lemma:propertiesOfLabeledParallelReduction}
\begin{enumerate}

\item If $A \superdevet{S,k} B$, then
    there exist variables $x_i$ and labels $a_i$, with $1 \leq i \leq k$,
    and $B' \in \termsEtiq$ s.t.
    $B = \lamm{a_1 \hdots a_k}{x_1 \hdots x_k}{B'}$.

\item If $A \superdevet{S,k} B$, then $\sinEtiq{A} \superdevelops{S,k}
  \sinEtiq{B}$.

\item If $M \superdevelops{S,k} N$, then there exist $M\et, N\et \in
  \termsEtiq$ s.t.  $M\et \superdevet{S,k}
  N\et$.
\end{enumerate}
\end{lemma}



\subsection{From Supersteps to Labeled Reduction}
\label{subsec:parallelToLabeled}


We address the proof of the first item of
Thm.~\ref{thm:equivalence}: If $M \superdevelops{S,0} N$,
  then there exist $M\et, N\et$ s.t.  $M\et \wreddDisj{S}
  N\et$.  Note that from $M \superdevelops{S,k} N$ and
Lem.~\ref{lemma:propertiesOfLabeledParallelReduction}(3), there exist
$M\et, N\et$ s.t.  $M\et \superdevet{S,k} N\et$. Thus we are left to
verify that 
\begin{equation}
M\et \superdevet{S,k} N\et \mbox{  implies }M\et
\wreddDisj{S} N\et
\label{eqn:parallelImpliesLabeled}
\end{equation}

In general, (\ref{eqn:parallelImpliesLabeled}) does not hold. The
problem is that $k > 0$ allows redexes with occurrences of bound
variables to be contracted. Eg. 
$\lamma{x}{\appm{b}{(\lamm{b}{y}{y})}{x}} \superdevet{\epsilon, 1}
\lamma{x}{x}$ but $\lamma{x}{\appm{b}{(\lamm{b}{y}{y})}{x}}
\not\wreddDisj{\epsilon} \lamma{x}{x}$. An attempt to prove
(\ref{eqn:parallelImpliesLabeled}) by induction on the derivation of
the judgement $M\et \superdevet{S,k} N\et$ reveals that we need to
consider a relaxed notion of labeled reduction in which contraction
of {\em some} redexes having free occurrences of bound variables is
admitted. The {\em only} such redexes which are allowed to be
contracted are those that contribute to the patterns of redexes that
are to be consumed in a weak superdevelopment. We call this {\em chain
  reduction}. Its definition arises from a fine analysis of the
generalization required of the hypothesis in order for the inductive
proof of (\ref{eqn:parallelImpliesLabeled}) to go through
(particularly when the term is an application).
\begin{definition}[Chain reduction]
\label{def:chainReduction}
The judgement $A \redEnc{S,k} B$ is defined by induction on $k$.
\begin{itemize}
\item $A \redEnc{S,0} B$ holds iff $A \wreddDisj{S} B$.

\item $A \redEnc{S,k+1} B$ holds iff there exist $A_1, A_2$ s.t. (1) $A \wreddDisj{S} \lamma{x}{A_1}$; (2) $A_1 \redEnc{S,k} A_2$; and (3) $B = \lamma{x}{A_2}$.
\end{itemize}
\end{definition}
Note that if $A \redEnc{S,k} B$, then
there exist variables $x_i$ and labels $a_i$, with $1 \leq i \leq k$, and $B' \in \termsEtiq$ s.t.
$B = \lamm{\overline{a}^k}{\overline{x}^k}{B'}$.

    The following
    congruence properties of chain reduction shall be required. The proof
    of those for application resort to
    Lem.~\ref{lem:weakeningOfBindingPathPreservesReduction},
    Lem.~\ref{lemaReduccionDeSustituciones}(2) and the fact that $A
    \wreddDisj{S} B$ implies $\freeVars{A} \setIntersection S =
    \freeVars{B} \setIntersection S$.

    \begin{lemma}[Abstraction]
    \label{lemaReduccionEnCadenaAbs}
    \begin{enumerate}
    \item If $B \redEnc{x \conc S,0} B'$, then
    $\lamma{x}{B} \redEnc{S,0} \lamma{x}{B'}$.

    \item If $B \redEnc{S,k} B'$, then $\lamma{x}{B} \redEnc{S,k+1} \lamma{x}{B'}$.
    \end{enumerate}
    \end{lemma}

    \begin{lemma}[Application I]
    \label{lemaReduccionEnCadenaApp1}
    If $A \redEnc{S,0} A'$ and $B \redEnc{S,0} B'$, then
    $\appma{A}{B} \redEnc{S,0} \appma{A'}{B'}$.
    \end{lemma}

    \begin{lemma}[Application II]
    \label{lemaReduccionEnCadenaApp2}
    \begin{enumerate}
    \item Let $A \redEnc{S,n+1} \lamma{x}{A'} = \lamm{a\,\overline{a}^n}{x\, \overline{x}^n}{x}$
    and $B \redEnc{S,m} \lamm{\overline{b}^m}{\overline{y}^m}{B'}$.
    Assume, moreover, $\lamma{x}{A'} \disj S$ and
    $B' \disj S$.  Then $\appma{A}{B} \redEnc{S,n+m} \lamm{\overline{a}^n \overline{b}^m}{\overline{x}^n \overline{y}^m}{B'}$.

    \item Let $A \redEnc{S,k+1} \lamma{x}{A'} = \lamm{a\, \overline{a}^k}{x\,\overline{x}^k}{A''}$
    and $B \redEnc{S,0} B'$.
    Assume, moreover,  $\lamma{x}{A'} \disj S$ and $B' \disj S$.
    Then $\appma{A}{B} \redEnc{S,k} \lamm{\overline{a}^k}{\overline{x}^k}{A''\sust{x}{B'}}$.

    \end{enumerate}
    \end{lemma}

    We can now replace (\ref{eqn:parallelImpliesLabeled}) by the following statement.

\begin{proposition}
\label{prop:headImpliesChainReduction}
    $A \superdevet{S,k} B$ implies $A \redEnc{S,k} B$.
\end{proposition}

    \begin{proof}
    By induction on the derivation of $A \superdevet{S,k} B$. Each case
    is straightforward, resorting to Lem.~\ref{lemaReduccionEnCadenaAbs}
    for the rules \laRegla{LAbs1} and \laRegla{LAbs2},
    Lem.~\ref{lemaReduccionEnCadenaApp1} for the rule \laRegla{LApp1}
    and Lem.~\ref{lemaReduccionEnCadenaApp2} for the rule \laRegla{LApp2}.
    \end{proof}

The proof of Thm.~\ref{thm:equivalence}(1) proceeds as follows.  From
$M \superdevelops{S,0} N$ and
Lem.~\ref{lemma:propertiesOfLabeledParallelReduction}(3), there exist
$M\et, N\et$ s.t.  $M\et \superdevet{S,0} N\et$. From
Prop.~\ref{prop:headImpliesChainReduction}(1), $M\et \redEnc{S,0}
N\et$. Finally, from the definition of chain reduction, $M\et
\wreddDisj{S} N\et$. The second item of
Prop.~\ref{prop:headImpliesChainReduction} will be used in the next
section.

\begin{corollary}
\label{coroPropertiesOfLabeledParallelReductionInclusion}
$\wredDisj{S}\ \subseteq\ \superdevet{S,0}\ \subseteq\ \wreddDisj{S}$.
\end{corollary}

    \begin{proof}
    \begin{enumerate}
    \item Suppose $A \wredDisj{S} B$. By induction on $A$ follows $A \superdevet{S,0} B$.
    \item Suppose $A \superdevet{S,0} B$.
    From Prop.~\ref{prop:headImpliesChainReduction}, $A \redEnc{S,0} B$.
    By Def.~\ref{def:chainReduction}, this implies $A \wreddDisj{S} B$.
    \end{enumerate}
    \end{proof}



\subsection{From Labeled Reduction to Supersteps}
\label{subsec:labeledToParallel}


A proof of Thm.~\ref{thm:equivalence}(2) requires reasoning over the
more general judgement $M\et \wreddDisj{S} B$ in which $B$ is not
necessarily in normal form but which is related to $N\et$ in the sense
that $B \wreddDisj{S} N\et$. In turn, for this it is convenient to
have a direct inductive characterization of the form of $N\et$ in
terms of $M\et$ (cf. Prop~\ref{prop:confluenceOfWeakLC}).  The
following notion of full weak superdevelopment provides such a
definition.  
\begin{definition}[Full weak superdevelopment]
The full weak superdevelopment of $A$ under $S,k$, written
$A\est{S,}$, is defined by induction on $A$ as follows:

\begin{center}
$\begin{array}{rcl}
    x\est{S,0} & \igdef & x\\
    (\lamma{x}{A})\est{S,0} & \igdef &  \lamma{x}{A\est{x\conc S,0}}\\
    (\lamma{x}{A})\est{S,k+1} & \igdef &  \lamma{x}{A\est{S,k}}\\
\appma{A}{B}\est{S,k} & \igdef &
  \left\{
     \begin{array}{ll}
       A'\sustLabel{a}\sust{x}{B'}, &  \text{if }(\star)\\
       \appma{A\est{S,0}}{B\est{S,0}},  & \text{if not }(\star)\text{ and }k = 0
     \end{array}
  \right.
\end{array}$
\end{center}
where $(\star)$ is the following condition: for some $n, m \geq 0$:
\begin{enumerate}
\item $A\est{S,n+1} = \lamma{x}{A'} = \lamm{a a_1 \hdots a_n}{x x_1 \hdots x_n}{A''}$
\item $B\est{S,m} = B' = \lamm{b_1 \hdots b_m}{y_1 \hdots y_m}{B''}$
\item $n + m = k$
\item $m > 0 \imp A'' = x$
\item $\appma{(\lamma{x}{A'})}{B'} \disj S$.

\end{enumerate}
\end{definition}

    \begin{remark}
      $A\est{S,k}$ may be undefined for $k > 0$. For example,
      $\appma{x}{y}\est{S,1}$ is undefined for any $S$ given that the full
      weak superdevelopment of $\appma{x}{y}$ does not produce an
      abstraction. Note, however, that $A\est{S,0}$ always exists.
      In general, if $A\est{S,k}$ is defined, then
        there exist variables
         $x_i$ and labels $a_i$, with $1 \leq i \leq k$, and $A' \in
        \termsEtiq$ s.t.
        $A\est{S,k} = \lamm{a_1 \hdots a_k}{x_1 \hdots x_k}{A'}$.
        Compare this with Lem.~\ref{lemma:propertiesOfLabeledParallelReduction}(1).
    \end{remark}

    \begin{lemma}
    \label{lemma:correctnessOfDefOfCompleteSuperdevelopment}
    If there exists an $A'$ such that $A\est{S,k}=A'$, it is unique.
    \end{lemma}

    Some basic properties of this notion follow. 

\begin{lemma}
\label{lemma:completeSuperDevAndReduction}
    If $A\est{S,k}$ exists:
\begin{enumerate}
\item 
$A \superdevet{S,k} A\est{S,k}$ when $A\est{S,k}$ exists.
\item $A \wreddDisj{S} A\est{S,0}$.

\end{enumerate}

\end{lemma}

    \begin{proof}
    The first item is proved by induction on $A$. The second follows from
    the first item and Cor.~\ref{coroPropertiesOfLabeledParallelReductionInclusion}.
    \end{proof}

    We now relate full weak superdevelopments with simultaneous
    labeled weak superdevelopments.

    \begin{lemma}
    \label{lemma:SuperdevsCoincidenEnSuperdevCompleto}

    \begin{enumerate}


    \item
    Let $A, B \in \termsEtiq$, $S \subseteq \conjVars$, $x \in \conjVars$
    and $k \geq 0$
    s.t. $x \not\in S \setIntersection \freeVars{A}$,
    $x \not\in \freeVars{B}$ and $B \disj S$.
    Then $(A\sust{x}{B})\est{S,k}$ exists iff the following conditions hold:
    \begin{enumerate}
    \item There exist $n, m \geq 0$ s.t. $k = n + m$,
    \item $A\est{S,n}$ and $B\est{S,m}$ exist,
    \item $m > 0$ implies $A\est{S,n}$ has the form $\lamm{a_1 \hdots a_n}{x_1 \hdots x_n}{x}$.
    \end{enumerate}

    Moreover, in that case,
    $(A\sust{x}{B})\est{S,k} = A\est{S,n}\sust{x}{B\est{S,m}}$.

    \item
    If
    $A, B \in \termsEtiq$, $S_1, S_2 \subseteq \conjVars$,
    $k_1, k_2 \geq 0$ s.t.
    $S_1 \supseteq S_2$, $k_1 \leq k_2$ and $A \superdevet{S_1,k_1} B$,
    then $A\est{S_2,k_2}$ exists iff $B\est{S_2,k_2}$ exists,
    and, if they exist, they coincide.

    \item Let $A_i \in \termsEtiq$ for $1 \leq i \leq n$, $S, K \subseteq
    \conjVars$ and $A_i \superdevet{S,k} A_{i+1}$ for all $1 \leq i <
    n$. Then $A_1\est{S,k}$ exists iff $A_n\est{S,k}$ exists,
    and, if they exist, they coincide.

    \end{enumerate}
    \end{lemma}

    \begin{proof}
    The first item is proved by induction on $A$ and resorting to the fact
    that
    $A\est{S,k} = A\est{x \conc S,k}$ if $x \not\in \freeVars{A}$.
    The second item
    is by induction on $A$ and resorts to the first one. The last one is a
    consequence of the second.
    \end{proof}

\begin{proposition}
\label{prop:confluenceOfWeakLC}
\weakLC is confluent.
\end{proposition}
    \begin{proof}
    Lem.~\ref{lemma:completeSuperDevAndReduction}(1) and
    Lem.~\ref{lemma:SuperdevsCoincidenEnSuperdevCompleto}(2) entail the
    diamond property of $\superdevet{S,0}$.
    This in turn entails confluence of \weakLC
    by Cor.~\ref{coroPropertiesOfLabeledParallelReductionInclusion}.
    \end{proof}

%

    \begin{lemma}
    \label{lemma:weakReductionToCompleteSuperdev}
    $A \wreddDisj{S} B$ implies $B \wreddDisj{S} A\est{S,0}$.
    \end{lemma}
    \begin{proof}
    Suppose $A \wreddDisj{S} B$.
    By Cor.~\ref{coroPropertiesOfLabeledParallelReductionInclusion}, $A \superdevett{S,0} B$.
    From Lem.~\ref{lemma:SuperdevsCoincidenEnSuperdevCompleto} we deduce
    $A\est{S,0} = B\est{S,0}$.
    Moreover, $B \wreddDisj{S} B\est{S,0} = A\est{S,0}$ from
    Lem.~\ref{lemma:completeSuperDevAndReduction}(2).
    \end{proof}

 Regarding the second item of Thm.~\ref{thm:equivalence},
  suppose $M\et \wreddDisj{S} N\et$ with $N\et$ in normal form.  From
  Lem.~\ref{lemma:weakReductionToCompleteSuperdev}, $N\et
  \wreddDisj{S} M\et\est{S,0}$.  Moreover, since $N\et$ is in normal
  form, it coincides with the complete weak superdevelopment of
  $M\et$, namely $N\et = M\et\est{S,0}$.  We conclude by resorting to
  Lem.~\ref{lemma:completeSuperDevAndReduction}(1) and then
  Lem.~\ref{lemma:propertiesOfLabeledParallelReduction}(2) to deduce
  $M \superdevelops{S,0} \sinEtiq{M\et\est{S,0}} = N$.


\section{Conclusions}

Redex creation in \weakLC is more subtle than in \LC. This raises the
question on how superdevelopments behave. We present two
possible definitions and prove that they are equivalent. The labeled
presentation is easy to grasp but complicated to use in proving
results (e.g. properties of reduction). Simultaneous reduction is easier
for this purpose. However, such an inductive definition requires
dealing with reduction under binders of redexes having free
occurrences of bound variables which labeled reduction forbids. This
makes the correspondence between these notions of reduction (labeled
and simultaneous) more demanding to prove, a task we have taken up in this
work.

We are currently developing these results in the framework of
higher-order rewriting (HOR). A number of issues arise in this
extended setting. First we must consider a notion of orthogonal HOR
systems for weak reduction, as discussed in the introduction. Second,
we have to determine what it means for a variable to be substituted in
order for reduction under binders of redexes involving these variables
to be allowed. Eg. in $\{ f (\lam{x}{g(y(x), z (x))}) \red{} f
(\lam{x}{g(y(a), z (x))}) \}$, redexes involving $x$ which occur below
$y$ (once this rule is instantiated) should be permitted but not those
below $z$. Last, there is an additional complication that is not
apparent in the setting of the lambda calculus. In the judgement $M
\superdevelops{S,n + 1} \lam{x}{M'}$ of \laRegla{App2} it turns out
that reduction steps involving free occurences of $x$ do not
contribute towards the creation of the outermost $\beta$ redex of
\laRegla{App2}. This requirement must be made explicit in HOR. The
naive generalization of $M \superdevelops{S,n + 1} \lam{x}{M'}$ to HOR
should require that the HOR-reduction steps involving free occurrences
of $x$ not contribute towards the newly created outermost redex.

\textbf{Acknowledgements.} To the referees for comments that helped
improve this paper.


\bibliographystyle{eptcs}
\bibliography{biblio}


\ignore{
\newpage
\appendix
\input{appendix}
}
\end{document}